\documentclass[runningheads]{llncs}
\usepackage{graphicx}
\usepackage{microtype}%if unwanted, comment out or use option "draft"
\usepackage{xspace}

\usepackage[end]{algpseudocode}
\usepackage{booktabs}
\usepackage{footnote}
\usepackage[normalem]{ulem}
\usepackage{tikz}
\usetikzlibrary{matrix,positioning,arrows.meta,arrows}
\usetikzlibrary{calc}
\usetikzlibrary{shapes.geometric,arrows}
\usetikzlibrary{decorations.text}
\usetikzlibrary{shapes}
\usetikzlibrary{automata,positioning}

\def\Oh{\mathcal{O}}
\long\def\ignore#1{}
\def\eqref#1{(\ref{#1})}
\newcommand{\xx}{\$}
\newcommand{\A}{\Sigma}
\newcommand{\Asize}{\sigma}

\newcommand{\LCP}{\mathsf{LCP}}
\newcommand{\LCS}{\mathsf{LCS}}

\newcommand{\nz}{{n_0}}
\newcommand{\none}{{n_1}}

\newcommand{\oneb}{{\bf 1}}
\newcommand{\zerob}{{\bf 0}}
\newcommand{\Bid}{\mathsf{Block\_id}}
\newcommand{\bid}{\mathsf{id}}

\newcommand{\bv}[1]{Z^{(#1)}}
\newcommand{\kh}{{b(h)}}
\newcommand{\sbot}{0}

\def\gSACA{{\sf gSACA-K}\xspace}

\algnewcommand\KwTo{\textbf{to }} 
\algnewcommand\KwAnd{\textbf{and }}
\algnewcommand\KwWrite{\textbf{write }}
\algnewcommand\KwTofile{\textbf{to file }}

\renewcommand{\tt}{\texttt}
\renewcommand{\u}{\underline}

\newcommand{\etal}{{\it et al.}\xspace}

\renewcommand{\S}{\mathcal{C}}

\newcommand{\dbG}{de Bruijn\xspace}
\newcommand{\repr}[1]{\overrightarrow{#1}}
\newcommand{\reprr}[1]{\overleftarrow{#1}}
\newcommand{\last}{\mathsf{last}}
\newcommand{\Wa}{{W}}
\newcommand{\Wx}{{W}^{-}}
\newcommand{\onex}{\mathbf{1}}
\newcommand{\zerox}{\mathbf{0}}
\newcommand{\F}{\mathsf{C}}
\newcommand{\Node}{\mathsf{Node}}

\newcommand{\LF}{LF}
\newcommand{\sta}{\mathsf{start}}

\newcommand{\Colz}{\mathcal{C}_0}
\newcommand{\Colo}{\mathcal{C}_1}
\newcommand{\Colzo}{\mathcal{C}_{01}}
\newcommand{\lastz}{\mathsf{last}_0}
\newcommand{\Waz}{{W_0}}
\newcommand{\Gz}{{G_0}}
\newcommand{\Wxz}{W_0^{-}}
\newcommand{\Wxo}{W_1^{-}}
\newcommand{\lasto}{\mathsf{last}_1}
\newcommand{\Wao}{{W_1}}
\newcommand{\Go}{{G_1}}
\newcommand{\lastx}[1]{\mathsf{last}_{#1}}
\newcommand{\Wax}[1]{{W_{#1}}}
\newcommand{\Wxx}[1]{W_{#1}^{-}}
\newcommand{\laz}[1]{\ell_0(#1)}
\newcommand{\lao}[1]{\ell_1(#1)}

\newcommand{\Bxx}{B_2}
\newcommand{\zzx}{\mathit{0}}
\newcommand{\oddx}{\mathit{1}}
\newcommand{\evenx}{\mathit{2}}
\newcommand{\oox}{\mathit{3}}

\newcommand{\edges}{m}

\def\mysubsubsection#1{\medbreak\noindent\textbf{#1.}}

\def\b#1{\mbox{\boldmath $#1$}}

\begin{document}
\title{Space-Efficient Merging of Succinct de~Bruijn~Graphs}
%\thanks{Supported by organization x.}
\author{Lavinia Egidi \inst{1} \and %\orcidID{0000-0002-9745-0942} \and
Felipe A. Louza\inst{2} \and %\orcidID{0000-0003-2931-1470} \and
Giovanni Manzini\inst{1,3} %\orcidID{0000-0002-5047-0196}
}
\authorrunning{L. Egidi et al.}
\institute{University of Eastern Piedmont, {Alessandria, Italy} \\
\email{\{lavinia.egidi, giovanni.manzini\}@uniupo.it}\\
\and
Department of Computing and Mathematics, University of S\~ao Paulo, Brazil
%{Ribeirão Preto, Brazil}\\
\email{louza@usp.br}\\
\and
IIT CNR, Pisa Italy
}
\maketitle              % typeset the header of the contribution
\begin{abstract}
We propose a new algorithm for merging succinct representations of {\em\dbG} graphs introduced in [Bowe {\em et al.} WABI 2012]. Our algorithm is based on the lightweight BWT merging approach by Holt and McMillan~[Bionformatics 2014, ACM-BCB 2014]. Our algorithm has the same asymptotic cost of the state of the art tool for the same problem presented by Muggli {\em et al.}~[bioRxiv 2017, Bioinformatics 2019], {but it uses less than half of its working space.} A novel {important} feature of our algorithm, not found in {any of} the existing tools, is that it can compute the {\em Variable Order} succinct representation of the union graph within the same asymptotic time/space bounds.

% starting from the plain representation of the input graphs.


\end{abstract}
\section{Introduction}

The {\em\dbG\ }graph for a collection of strings is a key data structure in genome assembly~\cite{Pevzner2001}. After the seminal work of Bowe \etal~\cite{wabi/BoweOSS12}, many succinct representations of this data structure have been proposed in the literature~\cite{wabi/AlmodaresiPP17,Belazzougui_2018,dcc/BoucherBGPS15,bioinformatics/MuggliBNMBRGPB17} offering more and more functionalities still using a fraction of the space required to store the input collection uncompressed. 
In this paper we consider the problem of merging two existing succinct representations of {\dbG\ graphs built for} different collections. Since the \dbG\ graph is a lossy representation and from it we cannot recover the original input collection, the alternative to  merging is storing a copy of each collection to be used for building new \dbG\ graphs from scratch. 

Recently, Muggli {\em et al.}~\cite{Muggli229641,MuggliBC19} have proposed a merging algorithm for colored \dbG\ graphs and have shown the effectiveness of the merging approach for the construction of \dbG\ graphs for very large datasets. The algorithm in~\cite{MuggliBC19} is based on an MSD Radix Sort procedure of the graph edges and its running time is $\Oh(m k)$, where $m$ is the total number of edges and $k$ is the order of the \dbG\ graph. 

A fundamental parameter of any construction algorithm for succinct data structures is its {\em space usage} since this parameter determines the size of the largest dataset that can be handled by a machine with a  given amount of memory. For a graph with $m$ edges and $n$ nodes the 
merging algorithm by Muggli \etal uses, in addition to the input and the output, $2(m\log\sigma+m+n)$ bits plus $\Oh(\sigma)$ words of working space, where $\sigma$ is the alphabet size. This value represents a three fold improvement over previous results, but it is still larger than the size of the resulting \dbG graph which is upper bounded by $2(m\log\sigma+m) + o(m)$ bits.

In this paper, we present a new merging algorithm that still  runs in $\Oh(m k)$ time, but only uses $4n$ bits  plus $\Oh(\sigma)$ words of working space. For genome collections ($\sigma=5$) our algorithm uses less than half the space of Muggli \etal's: our advantage grows with the size of the alphabet and with the average outdegree $m/n$. Notice that the working space of our algorithm is always less than the space of the resulting \dbG graph. In Section~\ref{s:implementation} we will discuss the practical significance of this space reduction.

Our new merging algorithm is based on a mixed {LSD/MSD  Radix Sort} algorithm which is inspired by the lightweight BWT merging algorithm introduced by Holt and McMillan~\cite{bcb/HoltM14,bioinformatics/HoltM14} and later improved in~\cite{spire/EgidiM17,GapjTA}. In addition to its small working space, our algorithm has the remarkable feature that it can compute as a by-product, with no additional cost, the $\LCS$ (Longest Common Suffix)  between the node labels, thus making it possible to construct succinct Variable Order \dbG\ graph representations~\cite{dcc/BoucherBGPS15}, a feature not shared by any other merging algorithm. 

The rest of the paper is organized as follows. After reviewing succinct \dbG graphs in Section~\ref{s:notation}, we describe our algorithm in Section~\ref{s:alg}. In Section~\ref{s:implementation} we describe the implementation details and compare our result to the state of the art. In Section~\ref{s:variants} we discuss the case of colored or variable order \dbG graphs. In Section~\ref{s:external} we show that combining an external memory version of our merging algorithm with recent results on external memory \dbG\ graph construction~\cite{wabi/EgidiLMT18,almob/EgidiLMT19} we get a space efficient external memory procedure for building succinct representations of \dbG\ graphs for very large collections.

\section{Notation and background}\label{s:notation}

Given the alphabet $\A = \{ 1,2,\ldots,\Asize\}$ and a collection of strings $\S = s_1, \ldots, s_d$ over $\A$, we prepend to each string $s_i$ $k$ copies of a symbol $\$ \notin \A$ which is lexicographically smaller than any other symbol. The order-$k$ {\it \dbG graph} $G(V,E)$ for the collection $\S$ is a directed edge-labeled graph containing a node $v$ for every {\bf unique $\b{k}$-mer} appearing in one of the strings of $\S$.
For each node $v$ we denote by $\repr{v} = v[1,k]$ its associated $k$-mer, where $v[1]\dots v[k]$ are symbols. 
The graph $G$ contains an edge $(u,v)$, with label $v[k]$, iff one of the strings in $\S$ contains a {\bf $\b{(k+1)}$-mer} with
prefix $\repr{u}$ and suffix $\repr{v}$.
The edge $(u,v)$ therefore represents the $(k+1)$-mer $u[1,k] v[k]$.
Note that each node has at most $|\A|$ outgoing edges and all edges incoming to node $v$ have label $v[k]$.

\mysubsubsection{BOSS succinct representation}
In 2012, Bowe {\em et al.}~\cite{wabi/BoweOSS12} introduced a succinct representation for the {de Bruijn} graph, usually referred to as BOSS representation, for the authors’ initials. The authors showed how to represent the graph in small space supporting fast navigation operations. The BOSS representation of the graph $G(V,E)$ is defined by considering the set of nodes $v_1, v_2, \ldots v_n$ sorted according to the colexicographic order of {their} associated {$k$-mer}. 
Hence, if $\reprr{v}=v[k]\dots v[1]$ denotes the string $\repr{v}$ reversed, the nodes are ordered so that 
\begin{equation}\label{eq:order}
\reprr{v_1} \prec \reprr{v_2} \prec \cdots \prec  \reprr{v_n}
\end{equation}
By construction the first node is $\reprr{v_1}=\$^k$ and all $\reprr{v_i}$ are distinct. 
For each node $v_i$, $i=1,\ldots,n$, we define $W_i$ as the sorted sequence of symbols on the edges leaving from node $v_i$; if $v_i$ has out-degree zero we set $W_i = \$$. 
Let $\Node[i]$ denote the node label for $W_i$.
Finally, we define
\begin{enumerate}
  \item $\Wa[1,m]$ as the concatenation $W_1 W_2 \cdots W_n$;
  \item $\Wx[1,m]$ as the bitvector such that $\Wx[i]=\oneb$ iff $\Wa[i]$ corresponds to the label of the edge $(u,v)$ such that $\reprr{u}$ has the smallest rank among the nodes that have an edge going to node $v$;
  \item $\last[1,m]$ as the bitvector such that $\last[i]=1$ iff $i=m$ or the outgoing edges corresponding to $\Wa[i]$ and $\Wa[i+1]$ have different source nodes.
  \item $\F[1,\sigma]$ as the integer array, such that $\F[c]$ stores the number of symbols smaller than $c \in \Sigma \cup \{\$\}$ in the last symbol of $\Node$. 
\end{enumerate}

The length $m$ of the arrays $\Wa$, $\Wx$, and $\last$ is equal to the number of edges plus the number of nodes with out-degree 0. In addition, the number of $\onex$'s in $\last$ is equal to the number of nodes $n$, and the number of $\onex$'s in $\Wx$ is equal to the number of nodes with positive in-degree, which is $n-1$ since $v_1=\$^k$ is the only node with in-degree 0.
Array $\F$ can be obtained by scanning $\Wa$, $\Wx$ and $\last$, therefore, array $\Node[1,m]$ is not stored explicitly.

Note that there is a natural one-to-one correspondence, called  $\LF$ for historical reasons, between the indices $i$ such that $\Wx[i]=\onex$ and the the set $\{2, \ldots,n\}$: in this correspondence $\LF(i)=j$ iff $v_j$ is the destination node of the edge associated to $\Wa[i]$.  See example in Figs.~\ref{f:dbg} and~\ref{f:boss}.

\begin{figure}[t]
    \centering
    \includegraphics{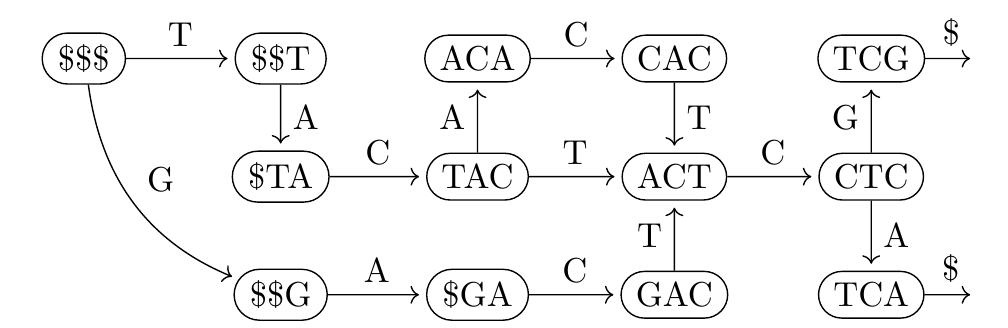}
    \caption{\dbG graph for $\S=\{$TACACT,
TACTCG, GACTCA$\}$.}
    \label{f:dbg}
\end{figure}

\begin{figure}[t]
    \centering
    \includegraphics[width=0.6\textwidth]{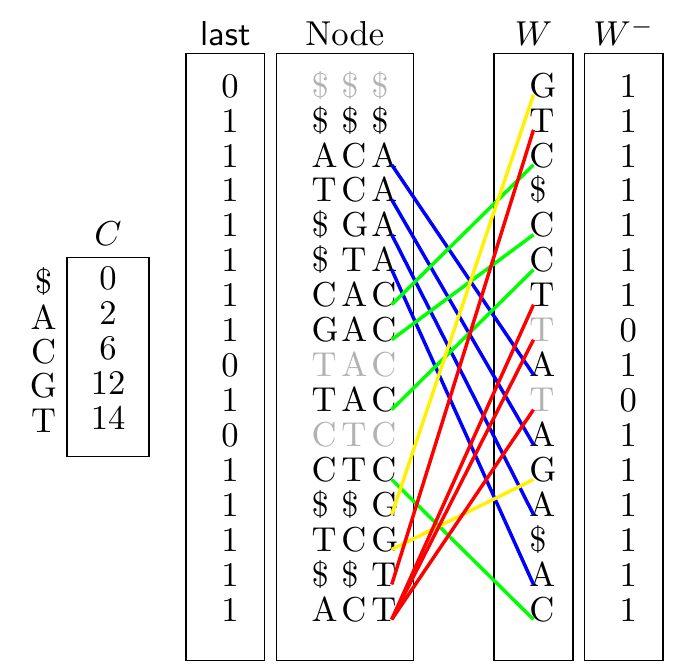}
    \caption{BOSS representation of the graph in Fig.~\ref{f:dbg}. The colored lines connect each label in $W$ to its destination node; edges of the same color have the same label. Note that edges of the same color do not cross because of Property~\ref{prop:lf}.}\label{f:boss}
\end{figure}

\begin{property}\label{prop:lf}
The $\LF$ map is order preserving in the following sense: if $\Wx[i]=\Wx[j]=\onex$ then
\begin{equation}\label{eq:wg}
\begin{array}{rcl}
\Wa[i] < \Wa[j]\; & \Longrightarrow\; &  \LF(i) < \LF(j),\\
(\Wa[i] = \Wa[j])\; \land (i < j)\; & \Longrightarrow\; & \LF(i) < \LF(j).
\end{array}
\end{equation}
\qed
\end{property}

In~\cite{wabi/BoweOSS12} it is shown that given array $\F$, enriching the arrays $\Wa$, $\Wx$, and $\last$ with the data structures from~\cite{FMMN04jou,RamanRS07} supporting constant time rank and select operations, we can efficiently navigate the graph $G$. 
The cost to store array $\F$ is $\Oh(\sigma \log n)$ bits.
The overall cost of encoding the three arrays and the auxiliary data structures is bounded by $m\log\sigma + 2m + o(m)$ bits, with the usual time/space tradeoffs available for rank/select data structures.

\mysubsubsection{Colored BOSS} 
The colored \dbG graph~\cite{iqbal2012a} is an extension of the \dbG graphs for a multiset of individual graphs, where each edge is associated with a set of ``colors'' that indicates which graphs contain that edge.

The BOSS representation for a set of graphs $\mathcal{G} = \{G_1, \dots, G_t\}$ contains the union of all individual graphs. In its simplest representation, 
the colors of all edges $W[i]$ are stored in a two-dimensional binary array $\mathcal{M}$, such that $\mathcal{M}[i,j]=1$ iff {the} $i$-th edge is present in graph $G_j$.
There are different compression alternatives for the color matrix $\mathcal{M}$ that support fast operations~\cite{wabi/AlmodaresiPP17,bioinformatics/MarcusLS14,bioinformatics/MuggliBNMBRGPB17}.
Recently, Alipanah {\em et~al.}~\cite{spire/AlipanahiKB18} presented a different approach to reduce the size of $\mathcal{M}$ by recoloring.

\mysubsubsection{Variable-order BOSS} 
The order $k$ (dimension) of a \dbG graph is an important parameter for genome assembling algorithms. 
The graph can be very small and uninformative when $k$ is small, whereas it can become too large or disconnected when $k$ is large. To add flexibility to the BOSS representation, Boucher {\em et al.}~\cite{dcc/BoucherBGPS15} suggest to enrich the BOSS representation of an order-$k$ \dbG graph with the length of the longest common suffix ($\LCS$) between the $k$-mers of consecutive nodes $v_1, v_2, \dots, v_n$ sorted according to~\eqref{eq:order}. These lengths are stored in a wavelet tree using $O(n \log k)$ additional bits. The authors show that this enriched representation supports navigation on all \dbG\ graphs of order $k'\leq k$ and that it is even possible to vary the order $k'$ of the graph on the fly during the navigation up to the maximum value $k$.

The $\LCS$ between $\repr{v_i}$ and $\repr{v_{i+1}}$ is equivalent to the length of the longest common prefix ($\LCP$) between their reverses $\reprr{v_i}$ and $\reprr{v_{i+1}}$. The $\LCP$ (or $\LCS$) between the nodes $v_1, v_2,  \cdots, v_n$ can be computed during the $k$-mer sorting phase. In the following we denote by VO-BOSS the {\bf variable order} succinct \dbG graph consisting of the BOSS representations enriched with the $\LCS/\LCP$ information.

\section{Merging plain BOSS representations}\label{s:alg}

Suppose we are given the BOSS {representations} of two \dbG graphs $\langle \Waz, \Wxz, \lastz \rangle$ and $\langle \Wao, \Wxo, \lasto \rangle$ obtained respectively from the collections of strings $\Colz$ and $\Colo$. In this section we show how to compute the BOSS representation for the union collection $\Colzo = \Colz \cup \Colo$. 
The procedure does not change in the general case when we are merging an arbitrary number of graphs. Let $\Gz$ and $\Go$ denote respectively the (uncompressed) \dbG graphs for $\Colz$ and $\Colo$, and let 
$$
v_1, \ldots, v_{\nz}\qquad\mbox{and}\qquad w_1, \ldots, w_{\none}
$$
denote their respective set of nodes sorted in colexicographic order. Hence, with the notation of the previous section we have
\begin{equation}\label{eq:sorted}
\reprr{v_1} \prec \cdots \prec \reprr{v_\nz} \qquad\mbox{and}\qquad 
\reprr{w_1} \prec \cdots \prec \reprr{w_\none}
\end{equation}
We observe that the $k$-mers in the collection $\Colzo$ are simply the union of the $k$-mers in $\Colz$ and $\Colo$. To build the \dbG graph for $\Colzo$ we need therefore to: 1) merge the nodes in $\Gz$ and $\Go$ according to the colexicographic order of their associated $k$-mers, 2) recognize when two nodes in $\Gz$ and $\Go$ refer to the same $k$-mer, and 3) properly merge and update the bitvectors $\Wxz$, $\lastz$ and $\Wxo$, $\lasto$.

\subsection{Phase 1: Merging $k$-mers} 

The main technical difficulty is that in the BOSS representation the $k$-mers associated to each node $\repr{v}=v[1,k]$ are not directly available. 
Our algorithm will reconstruct them using the symbols associated to the graph edges; to this end the algorithm will consider only the edges such that the corresponding {entries} in $\Wxz$ or $\Wxo$ {are} equal to $\onex$. 
Following these edges, first we recover the last symbol of each $k$-mer, following them a second time we recover the last two symbols of each $k$-mer and so on. However, to save space we do not explicitly maintain the $k$-mers; instead, using the ideas from~\cite{bcb/HoltM14,bioinformatics/HoltM14} our algorithm computes a bitvector $\bv{k}$ representing how the $k$-mers in $\Gz$ and $\Go$ should be merged according to the colexicographic order.  

To this end, our algorithm executes $k-1$ iterations of the code shown in Fig.~\ref{fig:xHMalgo} (note that lines~\ref{line:B=0h}--\ref{line:B=0hEnd} and~\ref{line:block_process_start}--\ref{line:block_process_end} of the algorithm are related to the computation of the $B$ array that is used in the following section). For $h=2,3,\ldots,k$, during iteration $h$, we compute the bitvector $\bv{h}[1,n_0+n_1]$ containing $n_0$ \zerob's and $n_1$ \oneb's such that $\bv{h}$ satisfies the following property

\begin{property}\label{prop:xhblock}
For $i=1,\ldots, \nz$ and $j=1,\ldots \none$ the $i$-th \zerob\ precedes the
$j$-th \oneb\ in $\bv{h}$ if and only if 
$\reprr{v_i}[1,h] \;\preceq\; \reprr{w_j}[1,h]$.
\qed
\end{property}

Property~\ref{prop:xhblock} states that if we merge the nodes from $\Gz$ and $\Go$ according to the bitvector $\bv{h}$ the corresponding $k$-mers will be sorted according to the lexicographic order restricted to the first $h$ symbols of each reversed $k$-mer. As a consequence, $\bv{k}$ will provide us the colexicographic order of all the nodes in $\Gz$ and $\Go$. 
To prove that Property~\ref{prop:xhblock} holds, we first define $\bv{1}$ and show that it satisfies the property, then we prove that for $h=2,\ldots,k$ the code in Fig.~\ref{fig:xHMalgo} computes $\bv{h}$ that still satisfies Property~\ref{prop:xhblock}. 

For $c\in\A$ let $\laz{c}$ and $\lao{c}$ denote respectively the number of nodes in $\Gz$ and $\Go$ whose associated $k$-mers end with symbol~$c$. These values can be computed with a single scan of $\Waz$ (resp. $\Wao$) considering only the symbols $\Waz[i]$ (resp. $\Wao[i]$) such that $\Wxz[i]=\onex$ (resp. $\Wxo[i]=\onex$). By construction, it is
$$
\nz = 1 + \sum_{c\in\A} \laz{c},\qquad\mbox\qquad \none = 1 + \sum_{c\in\A} \lao{c}
$$
where the two 1's account for the nodes $v_1$ and $w_1$ whose associated $k$-mer is $\xx^k$. 
We define 
\begin{equation}\label{eq:Z1}
\bv{1} = \u{\zerob\oneb}~ \u{\zerob^{\laz{1}}\oneb^{\lao{1}}}~ \u{\zerob^{\laz{2}}\oneb^{\lao{2}}} \cdots \u{\zerob^{\laz{\sigma}}\oneb^{\lao{\sigma}}}\;.
\end{equation}
The first pair \zerob\oneb\ in $\bv{1}$ accounts for $v_1$ and $w_1$; for each $c\in\A$ group $\zerob^{\laz{c}}\oneb^{\lao{c}}$ accounts for the nodes ending with symbol~$c$. Note that, apart from the first two symbols, $\bv{1}$ can be logically partitioned into $\sigma$ subarrays one for each alphabet symbol. For $c\in\A$ let 
$$\sta(c) = 3 + \sum_{i<c}(\laz{i} + \lao{i})$$
then the subarray corresponding to $c$ starts at position $\sta(c)$ and has size $\laz{c} + \lao{c}$. 
As a consequence of~\eqref{eq:sorted}, the $i$-th \zerob\ (resp. $j$-th \oneb) belongs to the subarray associated to symbol $c$ iff $\reprr{v_i}[1]=c$ (resp. $\reprr{w_j}[1]=c$). 

To see that $\bv{1}$ satisfies Property~\ref{prop:xhblock}, observe that the $i$-th \zerob\ precedes $j$-th \oneb\ iff the $i$-th \zerob\ belongs to a subarray corresponding to a symbol not larger than the symbol corresponding to the subarray containing the $j$-th \oneb; this implies $\reprr{v_i}[1,1] \preceq \reprr{w_j}[1,1]$. 

The bitvectors $\bv{h}$ computed by the algorithm in Fig.~\ref{fig:xHMalgo} can be logically divided into the same subarrays we defined for $\bv{1}$. In the algorithm we use an array $F[1,\sigma]$ to keep track of the next available position of each subarray. Because of how the array $F$ is initialized and updated, we see that every time we read a symbol $c$ at line~\ref{step:getc} the corresponding bit $b=\bv{h-1}[k]$, which gives us the graph containing~$c$, is written in the portion of $\bv{h}$ corresponding to $c$ (line~\ref{step:putc}). The only exception are the first two entries of $\bv{h}$ which are written at line~\ref{step:01} which corresponds to the nodes $v_1$ and $w_1$. We treat these nodes differently since they are the only ones with in-degree zero. For all other nodes, we implicitly use the one-to-one correspondence {\eqref{eq:wg}} between entries $W[i]$  with $\Wx[i]=\oneb$ and nodes $v_j$ with positive in-degree.

The following Lemma proves the correctness of the algorithm in Fig.~\ref{fig:xHMalgo}.

\begin{lemma}\label{lemma:xhblock}
For $h=2,\ldots,k$, the array $\bv{h}$ computed by the algorithm in Fig.~\ref{fig:xHMalgo} satisfies Property~\ref{prop:xhblock}.
\end{lemma}

\begin{proof}
To prove the ``if'' part of Property~\ref{prop:xhblock} let $1 \leq f < g \leq \nz+\none$ denote two indexes such that {$\bv{h}[f]$} is the $i$-th \zerob\ and {$\bv{h}[g]$} is the $j$-th \oneb\ in $\bv{h}$ for some $1 \leq i \leq \nz$ and $1 \leq j \leq \none$. We need to show that $\reprr{v_i}[1,h] \preceq \reprr{w_j}[1,h]$. 

Assume first $\reprr{v_i}[1]\neq \reprr{w_j}[1]$. The hypothesis $f<g$ implies $\reprr{v_i}[1]<\reprr{w_j}[1]$, since otherwise during iteration~$h$ the $j$-th \oneb\ would have been written in a subarray of $\bv{h}$ preceding the one where the $i$-th \zerob\ is written. Hence $\reprr{v_i}[1,h] \preceq \reprr{w_j}[1,h]$ as claimed. 
 
Assume now $\reprr{v_i}[1] = \reprr{w_j}[1] = c$. In this case during iteration $h$
the $i$-th \zerob\ and the $j$-th \oneb\ are both written to the subarray of $\bv{h}$ associated to symbol~$c$. Let $f'$, $g'$ denote respectively the value of the main loop variable~$p$ in the procedure of Fig.~\ref{fig:xHMalgo} when the entries {$\bv{h}[f]$} and {$\bv{h}[g]$} are written. Since each subarray in $\bv{h}$ is filled sequentially, the hypothesis $f<g$ implies $f'<g'$. By construction $\bv{h-1}[{f}']=\zerox$ and $\bv{h-1}[{g}']=\onex$. Say ${f}'$ is the $i'$-th \zerob\ in $\bv{h-1}$ and ${g}'$ is the $j'$-th \oneb\ in $\bv{h-1}$. By the inductive hypothesis on $\bv{h-1}$ it is
\begin{equation}\label{eq:xhblock2}
\reprr{v_{i'}}[1,h-1] \;\preceq\; \reprr{w_{j'}}[1,h-1].
\end{equation}
By construction there is an edge labeled $c$ from $v_{i'}$ to $v_i$ and from $w_{j'}$ to $w_j$ hence 
$$
\repr{v_{i}}[1,h] = \repr{v_{i'}}[1,h-1]c,\qquad \repr{w_{j}}[1,h] = \repr{w_{j'}}[1,h-1]c;
$$
therefore
$$
\reprr{v_{i}}[1,h] = c \reprr{v_{i'}}[1,h-1],\qquad \reprr{w_{j}}[1,h] = c \reprr{w_{j'}}[1,h-1];
$$
using~\eqref{eq:xhblock2} we conclude that $\reprr{v_i}[1,h] \preceq \reprr{w_j}[1,h]$ as claimed.

For the ``only if'' part of Property~\ref{prop:xhblock}, assume $\reprr{v_i}[1,h] \preceq \reprr{w_j}[1,h]$ for some $i\geq 1$ and $j\geq 1$. We need to prove that in $\bv{h}$ the $i$-th \zerob\ precedes the $j$-th \oneb. If $\reprr{v_i}[1]\neq\reprr{w_j}[1]$ the proof is immediate. If $c=\reprr{v_i}[1]=\reprr{w_j}[1]$ then
$$
\reprr{v_i}[2,h]\preceq\reprr{w_j}[2,h].
$$
Let $i'$ and $j'$ be such that $\reprr{v_{i'}}[1,h-1] = \reprr{v_i}[2,h]$ and $\reprr{w_{j'}}[1,h-1] =\reprr{w_j}[2,h]$.
By induction {hypothesis}, in $\bv{h-1}$ the $i'$-th \zerob\ precedes the $j'$-th \oneb. 

During phase~$h$, the $i$-th \zerob\ in $\bv{h}$ is written to position $f$ when processing the $i'$-th \zerob\ of $\bv{h-1}$, and the $j$-th \oneb\ in $\bv{h}$ is written to position $g$ when processing the $j'$-th \oneb\ of $\bv{h-1}$. Since in $\bv{h-1}$ the $i'$-th \zerob\ precedes the $j'$-th \oneb\ and since $f$ and $g$ both belong to the subarray of $\bv{h}$ corresponding to the symbol $c$, their relative order does not change and the $i$-th \zerob\ precedes the $j$-th \oneb\ as claimed.\qed
\end{proof}

\begin{figure}[t]
\hrule\smallbreak
\begin{algorithmic}[1]
\For{$c \gets 1$ \KwTo $\sigma$} \label{line:initloop}
  \State $F[c] \gets \sta(c)$    \Comment{Init $F$ array}
  \State $\Bid[c] \gets -1$      \Comment{Init $\Bid$ array}
\EndFor
\State $i_0 \gets i_1 \gets 1$  \Comment{Init counters for $\Waz$ and $\Wao$}
\State $\bv{h} \gets \zerob\oneb$  \Comment{First two entries correspond to {$v_1$ and $w_1$}}\label{step:01}
\For{$p \gets 1$ \KwTo $n_0+n_1$}\label{line:initmainloop}
    \If{$B[p]\neq \sbot $ \KwAnd $B[p]\neq h$}\label{line:B=0h}
       \State $\bid\gets p$\Comment{A new block of $\bv{h-1}$ is starting}\label{line:blockstart}
    \EndIf  \label{line:B=0hEnd}
    \State $b \gets \bv{h-1}[p]$\Comment{Get bit $b$ from $\bv{h-1}$}
    \Repeat{}\Comment{Current node is from graph $G_b$}
      \If {$\Wxx{b}[i_b]=\oneb$}\label{line:W=1}
        \State $c \gets \Wax{b}[i_b]$ \Comment{Get symbol from outgoing edges}\label{step:getc}
        \State $q \gets F[c]{\mathsf ++}$ \Comment{Get destination for $b$ according to symbol $c$}
        \State $\bv{h}[q] \gets b$        \Comment{Copy bit $b$ to $\bv{h}$}\label{line:updatebv}\label{step:putc}
         \If{$\Bid[c]\neq \bid$}\label{line:block_process_start}
          \State $\Bid[c]\gets \bid$\Comment{Update block id for symbol $c$}
          \If{$B[q] = \sbot$} \Comment{Check if already marked}\label{line:new_start}
            \State$B[q] \gets h$\Comment{A new block of $\bv{h}$ will start here}\label{line:writeh}
          \EndIf
         \EndIf \label{line:block_process_end}        
      \EndIf 
    \Until{$\lastx{b}[i_b{\mathsf ++}] \neq \onex$} \Comment{Exit if $c$ was last edge}
\EndFor\label{line:endmainloop}
\end{algorithmic}
\smallbreak\hrule
\caption{Main procedure for merging succinct \dbG graphs. Lines~\ref{line:B=0h}--\ref{line:B=0hEnd} and~\ref{line:block_process_start}--\ref{line:block_process_end} are related to the computation of the $B$ array introduced in Section~\ref{s:phase2}.}\label{fig:xHMalgo}
\end{figure}

\subsection{Phase 2: Recognizing identical $k$-mers} \label{s:phase2}
Once we have determined, via the bitvector $\bv{h}[1, n_0+n_1]$, the colexicographic order of the $k$-mers, we need to determine when two $k$-mers are identical since in this case we have to merge their outgoing and incoming edges. Note that two identical $k$-mers will be consecutive in the colexicographic order and they will necessarily belong one to $\Gz$ and the other to $\Go$.

Following Property~\ref{prop:xhblock}, and a technique introduced in~\cite{spire/EgidiM17}, we identify the $i$-th \zerob\ in
$\bv{h}$ with $\reprr{v_i}$ and the $j$-th \oneb\ in $\bv{h}$ with $\reprr{w_j}$.
Property~\ref{prop:xhblock} is equivalent to state that we can logically
partition $\bv{h}$ into $\kh+1$ $h$-blocks
\begin{equation}\label{eq:Zblocks}
\bv{h}[1,\ell_1],\; \bv{h}[\ell_1+1, \ell_2],\; \ldots,\;
\bv{h}[\ell_\kh+1,n_0+n_1]
\end{equation}
such that each block corresponds to a set of $k$-mers which are prefixed by the same length-$h$ substring. 
Note that during iterations {$h=2,3,\dots,k$} the $k$-mers within an $h$-block will be rearranged, and sorted according to longer and longer prefixes, but they will stay within the same block. 

In the algorithm of Fig.~\ref{fig:xHMalgo}, in addition to $\bv{h}$, we maintain an integer array $B[1,\nz+\none]$, such that at the end of iteration~$h$ it is $B[i]\neq 0$ if and only if a block of $\bv{h}$ starts at position~$i$. Initially, for $h=1$, since we have one block {per} symbol, we set
$$
B=\u{1 0}\, \u{1 0^{\laz{1}+\lao{1}-1}}\, \u{1 0^{\laz{2}+\lao{2}-1}} \cdots \u{10^{\laz{\sigma}+\lao{\sigma}-1}}.
$$
During iteration~$h$, new block boundaries are established as follows. At line~\ref{line:blockstart} we identify each existing block with its starting position. Then, at lines~\ref{line:block_process_start}--\ref{line:block_process_end}, if the entry $\bv{h}[q]$ has the form $c\alpha$, while $\bv{h}[q-1]$ has the form $c\beta$, with $\alpha$ and $\beta$ belonging to different blocks, then we know that $q$ is the starting position of an $h$-block. Note that we write $h$ to $B[q]$ only if no other value has been previously written there. This ensures that $B[q]$ is the smallest position in which the strings corresponding to $\bv{h}[q-1]$ and $\bv{h}[q]$ differ, or equivalently, $B[q]-1$ is the LCP between the strings corresponding to $\bv{h}[q-1]$ and $\bv{h}[q]$. The above observations are summarized in the following Lemma, which is a generalization to \dbG\ graphs of an analogous result for BWT merging established in Corollary~4 in~\cite{spire/EgidiM17}.

\begin{lemma}\label{lemma:lcp}
After iteration~$k$ of the merging algorithm for $q=2,\ldots, \nz+\none$ if $B[q]\neq 0$ then $B[q]-1$ is the LCP between the {reverse} $k$-mers corresponding to $\bv{k}[q-1]$ and $\bv{k}[q]$, while if $B[q]=0$ {their LCP is equal to $k$}, hence such $k$-mers are equal.\qed
\end{lemma}

The above lemma shows that using array $B$ we can establish when two $k$-mers are equal and consequently the associated graph nodes should be merged. 

\subsection{Phase 3: Building BOSS representation for the union graph}
We now show how to compute the succinct representation of the union graph $\Gz~\cup~\Go$, consisting of the arrays $\langle \Wax{01}$, $\Wxx{01}$, $\lastx{01}\rangle$, given the succinct representations of $\Gz$ and $\Go$ and the arrays $\bv{k}$ and $B$.

The arrays  $\Wax{01}$, $\Wxx{01}$, $\lastx{01}$ are initially empty and we fill them in a single {sequential} pass. 
For $q=1,\ldots,\nz+\none$ we consider the values $\bv{k}[q]$ and $B[q]$.
If $B[q]=0$ then the $k$-mer associated to $\bv{k}[q-1]$, say $\reprr{v_i}$ is identical to the $k$-mer associated to  $\bv{k}[q]$, say $\reprr{w_j}$. In this case we recover from $\Wax{0}$ and $\Wax{1}$ the labels of the edges outgoing from $v_i$ and $w_j$, we compute their union and write them to $\Wax{01}$ {(we assume the edges are in the lexicographic order)}, writing at the same time the 
representation of the out-degree of the new node to $\lastx{01}$. 
If instead $B[q]\neq 0$, then the $k$-mer associated to $\bv{k}[q-1]$ is unique and we copy the information of its outgoing edges {and out-degree} directly to $\Wax{01}$ and $\lastx{01}$.

When we write the symbol $\Wax{01}[i]$ we simultaneously write the bit $\Wxx{01}[i]$ according to the following strategy. 
If the symbol $c=\Wax{01}[i]$ is the first occurrence of $c$ after a value $B[q]$, with $0 < B[q] < k$, then we set $\Wxx{01}[i]=\oneb$, otherwise we set $\Wxx{01}[i]=\zerob$.
The rationale is that if no values $B[q]$ with $0 < B[q] < k$  occur between two nodes, then the associated (reversed) $k$-mers have a common LCP of length $k-1$ and therefore if they both have an outgoing edge labelled with $c$ they reach the same node and only the first one  should have $\Wxx{01}[i]=\oneb$.

\section{Implementation details and analysis}\label{s:implementation}

Let $n=\none+\nz$ denote the sum of number of nodes in $\Gz$ and $\Go$, and let $\edges=|\Waz|+|\Wao|$ denote the sum of the number of edges. The $k$-mer merging algorithm as described executes in $\Oh(\edges)$ time a first pass over the arrays $\Waz$, $\Wxz$, and $\Wao$, $\Wxo$ to compute the values $\laz{c} + \lao{c}$ for $c\in\A$ and initialize the arrays $F[1,\sigma]$, $\sta[1,\sigma]$, $\Bid[1,\sigma]$ and $\bv{1}[1,n]$ (Phase 1). 
Then, the algorithm executes $k-1$ iterations of the code in Fig.~\ref{fig:xHMalgo} each iteration taking $\Oh(\edges)$ time. Finally, still in $\Oh(\edges)$ time the algorithm computes the succinct representation of the union graph (Phases 2 and 3). The overall running time is therefore $\Oh(\edges\, k)$.

We now analyze the space usage of the algorithm.  In addition to the input and the output, our algorithm uses $2n$ bits for two instances of the $\bv{\cdot}$ array (for the current $\bv{h}$ and for the previous $\bv{h-1}$), plus $n\lceil\log k\rceil$ bits for the $B$ array. 
Note, however, that during iteration $h$ we only need to check whether $B[i]$ is equal to 0, $h$, or some value within 0 and $h$. Similarly, for the computation of $\Wxx{01}$ we only need to distinguish between the cases where $B[i]$ is equal to 0, $k$ or some value $0 < B[i]< k$.
Therefore, we can save space replacing $B[1,n]$ with an array $\Bxx[1,n]$ containing two bits per entry representing the four possible states $\{\zzx,\oddx,\evenx,\oox\}$. During iteration $h$, the values in $\Bxx$ are used instead of the ones in $B$ as follows: An entry $\Bxx[i]=\zzx$ corresponds to $B[i]=0$, an entry $\Bxx[i]=\oox$ corresponds to an entry $0 < B[i] < h-1$. In addition, if $h$ is even, an entry $\Bxx[i]=\evenx$ corresponds to $B[i]=h$ and an entry $\Bxx[i]=\oddx$ corresponds to $B[i]=h-1$; while if $h$ is odd the correspondence is $\evenx \rightarrow h-1$, $\oddx \rightarrow h$. The reason for this apparently involved scheme, first introduced in~\cite{wabi/EgidiLMT18}, is that during phase $h$, an entry in $\Bxx$ can be modified either before or after we have read it at Line~\ref{line:blockstart}. Using this technique, the working space of the algorithm, i.e., the space in addition to the input and the output, is $4n$ bits plus $3\sigma + \Oh(1)$ words of RAM for the arrays $\sta$, $F$, and $\Bid$.

\begin{theorem}\label{t:merge1}
The merging of two succinct representations of two order-$k$ \dbG\ graphs can be done in $\Oh(\edges\, k)$ time using $4n$ bits plus $\Oh(\sigma)$ words of working space.\qed
\end{theorem}

We stated the above theorem in terms of working space, since the total space depends on how we store the input and output, and for such storage there are several possible alternatives. The usual assumption is that the input \dbG graphs, i.e. the arrays $\langle \Waz, \Wxz, \lastz \rangle$ and $\langle \Wao, \Wxo, \lasto \rangle$, are stored in RAM using overall $m\log \sigma + 2m$ bits. Since the three arrays representing the output \dbG graph are generated sequentially in one pass, they are usually written directly to disk without being stored in RAM, so they do not contribute to the total space usage. Also note that during each iteration of the algorithm in Fig.~\ref{fig:xHMalgo}, the input arrays are all accessed sequentially. Thus we could keep them on disk reducing the overall RAM usage to just $4n$ bits plus $\Oh(\sigma)$ words; the resulting algorithm would perform additional $\Oh( k(m\log \sigma + 2m)/D )$ I/Os where $D$ denotes the disk page size in bits.

\mysubsubsection{Comparison with the state of the art} The \dbG graph merging algorithm by Muggli {\em et al.}~\cite{MuggliBC19,Muggli229641} is similar to ours in that it has a {\em planning phase} consisting of the colexicographic sorting of the $(k+1)$-mers associated to the edges of $G_0$ and $G_1$. To this end, the algorithm uses a standard MSD radix sort. However only the most significant symbol of each $(k+1)$-mer is readily available in $\Waz$ and $\Wao$. Thus, during each iteration the algorithm computes also the next symbol of each $(k+1)$-mer that will be used as a sorting key in the next iteration. The overall space for such symbols is $2m\lceil \log \sigma\rceil$ bits, since for each edge we need the symbol for the current and next iteration. In addition, the algorithm {uses up to $2(n+m)$ bits}
to maintain the set of intervals consisting in edges whose associated reversed $(k+1)$-mer have a common prefix; these intervals correspond to the blocks we implicitly maintain in the array $\Bxx$ using only $2n$ bits.

Summing up, the algorithm by Muggli {\em et al.} runs in $\Oh(mk)$ time, and uses $2(m\lceil\log\sigma\rceil + m + n)$ bits plus $\Oh(\sigma)$ words of working space. Our algorithm has the same time complexity but uses less space: even for $\sigma=5$ as in bioinformatics applications, our algorithm uses less than half the space ($4n$ bits vs. $6.64 m+2n$ bits).
This space reduction significantly influences the size of the largest \dbG graph that can be built with a given amount of RAM. For example, in the setting in which the input graphs are stored on disk and all the RAM is used for the working space, our algorithm can build a \dbG graph whose size is twice the size of the largest \dbG graph that can be built with the algorithm of Muggli \etal.

We stress that the space reduction was obtained by substantially changing the sorting procedure. Although both algorithms are based on radix sorting they differ substantially in their execution. The algorithm by Muggli {\em et al.} follows the traditional MSD radix sort strategy; hence it establishes, for example, that $ACG \prec ACT$ when it compares the third `digits` and finds that $G < T$. In our algorithm we use a mixed LSD/MSD strategy: in the above example we also find that $ACG \prec ACT$ during the third iteration, but this is established without comparing directly $G$ and $T$, which are not explicitly available. Instead, during the second iteration the algorithm finds that $CG \prec CT$ and during the third iteration it uses this fact to infer that $ACG \prec ACT$: this is indeed a remarkable sorting trick first introduced in~\cite{bioinformatics/HoltM14} and adapted here to \dbG graphs.

\section{Merging colored and VO-BOSS representations}\label{s:variants}

Our algorithm can be easily generalized to merge colored and VO (variable-order) BOSS representations. Note that the algorithm by Muggli {\em et al.} can also merge colored BOSS representations, but in its original formulation, it cannot merge VO representations.

Given the colored BOSS representation of two \dbG graphs $\Gz$ and $\Go$, the corresponding color matrices $\mathcal{M}_0$ and $\mathcal{M}_1$ have size $m_0 \times c_0$ and $m_1 \times c_1$. 
We initially create a new color matrix $\mathcal{M}_{01}$ of size $(m_0+m_1) \times (c_0+c_1)$ with all entries empty.
During the merging {of the union graph (Phase 3)}, for $q=1,\ldots,n$, we write the colors of the edges associated to $\bv{h}[q]$ to the corresponding line in $\mathcal{M}_{01}$  {possibly merging the colors when we find nodes with identical $k$-mers}  {in $\Oh(c_{01})$ time, with $c_{01}=c_0+c_1$.} 
To make sure that color {\sf id}s from $\mathcal{M}_{0}$ are different from those in $\mathcal{M}_{1}$  in the new graph we add the constant $c_0$ (the number of distinct colors in $\Gz$) to any color id coming from the matrix $\mathcal{M}_1$.  

\begin{theorem}
The merging of two succinct representations of colored \dbG\ graphs takes $\Oh(\edges \, \max(k,c_{01}))$ time and $4n$ bits  plus $\Oh(\sigma)$ words of working space, where $c_{01} = c_0+c_1$. \qed
\end{theorem}

We now show that we can compute the variable order VO-BOSS representation of the union of two \dbG graphs $G_0$ and $G_1$ given their {\em plain}, eg. non variable order, BOSS representations. 
For the VO-BOSS representation we need the $\LCS$ array for the nodes in the union graph $\langle \Wax{01}$, $\Wxx{01}$, $\lastx{01}\rangle$.
Notice that after merging the $k$-mers of $\Gz$ and $\Go$ with the algorithm in Fig.~\ref{fig:xHMalgo} (Phase~1) the values in $B[1,n]$ already provide the LCP information between the reverse labels of all consecutive nodes (Lemma~\ref{lemma:lcp}).
When building the union graph (Phase~3), for $q=1,\ldots,n$, the $\LCS$ between two consecutive nodes, say $v_i$ and $w_j$, is equal to the $\LCP$ of their reverses $\reprr{v_i}$ and $\reprr{w_j}$, which is given by $B[q]-1$ whenever $B[q]>0$ (if $B[q]=0$ then $\reprr{v_i}=\reprr{w_j}$ and nodes $v_i$ and $v_j$ should be merged).  Hence, our algorithm for computing the VO representation of the union graph consists exactly of the algorithm in Fig.~\ref{fig:xHMalgo} in which we store the array $B$ in $n\log k$ bits instead of using the 2-bit representation described in Section~\ref{s:implementation}. Hence the running time is still $\Oh(m k)$ and the  working space becomes the space for the bitvectors $\bv{h-1}$ and $\bv{h}$ (recall we define the working space as the space used in addition to the space for the input and the output). 

\begin{theorem}
Merging two succinct representations of variable order \dbG\ graphs takes $\Oh(\edges k)$ time and $2n$ bits plus $\Oh(\sigma)$ words of working space.\qed
\end{theorem}

\section{External memory construction}\label{s:external}

In this section we show that using our merging algorithm we can design a complete external memory algorithm to construct succinct \dbG graphs.

We preliminary observe that at each iteration of the algorithm in Fig.~\ref{fig:xHMalgo} not only the arrays $\langle \Waz, \Wxz, \lastz \rangle$ and $\langle \Wao, \Wxo, \lasto \rangle$ but also $\bv{h-1}$ and $\Bxx$ are read sequentially from beginning to end. At the same time, the arrays $\bv{h}$ and $\Bxx$ are written sequentially but into $\sigma$ different partitions whose starting positions are the values in~$\sta[1,\sigma]$ which are the same for each iteration. Thus, if we split $\bv{\cdot}$ and $\Bxx$ into $\sigma$ different files, all accesses are sequential and our algorithm runs in external memory in $\Oh(mk)$ time, doing $\Oh(mk)$ sequential I/Os and using only $\Oh(\sigma)$ words of RAM.

Assume now we are given a string collection $\S = s_1, \ldots, s_d$ of total length $N$, the desired order $k$, and the amount of available RAM $M$. First, we split $\S$ into smaller subcollections $r_i = s_j,\dots,s_{j'}$, such that we can compute the BWT and LCP array of each subcollection in linear time in RAM using $M$ bytes, using {\em e.g.} the suffix sorting algorithm~\gSACA~\cite{tcs/LouzaGT17}. For each subcollection we then compute, and write to disk, the BOSS representation of its \dbG graph using the algorithm described in~\protect{\cite[Section 5.3]{wabi/EgidiLMT18}}.
Since these are linear algorithms the overall cost of this phase is $\Oh(N)$ time and $\Oh(N)$ sequential I/Os. 

Finally, we merge all \dbG graphs into a single BOSS representation of the union graph with the external memory variant just described. Since the number of subcollections is $\Oh(N/M)$, a total of $\log(N/M)$ merging rounds will suffice to get the BOSS representation of the union graph. 

\begin{theorem}
Given a strings collection $\S = s_1, \ldots, s_d$ of total length $N$, we can build the corresponding order-$k$ succinct \dbG graph in $\Oh(N\,k\log(N/M))$ time and $\Oh(N\,k\log(N/M))$ sequential I/Os using $\Oh(M)$ words of RAM.\qed
\end{theorem}

Note that our construction algorithm can be easily extended to generate the colored/variable order variants of the \dbG graph. For the colored variant it suffices to use \gSACA\ to generate also the document array~\cite{tcs/LouzaGT17} and then use the colored merging variant. For the variable order representation, it suffices to store the $\LCP/\LCS$ values during the very last merging phase, using the techniques described in~\protect{\cite[Section~3]{wabi/EgidiLMT18}} to handle them in external memory.

\section*{Acknowledgments}

\paragraph{Funding.}
L.E. and G.M. were partially supported by PRIN grant 2017WR7SHH. L.E. was partially supported by the University of Eastern Piedmont project  {\sl Behavioural Types for Dependability Analysis with Bayesian Networks}.
F.A.L. was supported by the grants $\#$2017/09105-0 and $\#$2018/21509-2 from the S\~ao Paulo Research Foundation (FAPESP).
G.M. was partially supported by INdAM-GNCS Project 2019 {\sl Innovative methods for the solution of medical and biological big data} and by the  LSBC\_19-21 Project from the University of Eastern Piedmont.


\end{document}